\documentclass{article}
\usepackage{amsmath}
\usepackage{amsfonts}
\usepackage{amssymb}
\usepackage{amsthm}

\def\p{\partial}

\newtheorem{theorem}{Theorem}

\newtheorem*{corollary}{Corollary}
\newcommand{\dbar}{\bar{\partial}}
\newcommand{\wt}{\widetilde}
\newcommand{\be}{\begin{equation}}
\newcommand{\ee}{\end{equation}}
\newcommand{\bea}{\begin{eqnarray}}
\newcommand{\eea}{\end{eqnarray}}
\newcommand{\beaa}{\begin{eqnarray*}}
\newcommand{\eeaa}{\end{eqnarray*}}

\newcommand{\nn}{\nonumber}
\renewcommand{\d}{\mathrm{d}}

\begin{document}
\title{SDYM equations on the self-dual background}
\author{
L.V. Bogdanov\thanks
{L.D. Landau ITP RAS,
Moscow, Russia}}
\date{}
\maketitle
\begin{abstract} 
We introduce the technique combining the features of integration schemes for SDYM equations and multidimensional dispersionless integrable equations to get SDYM equations on the conformally self-dual background. Generating differential form is defined, the dressing scheme is developed. Some special cases and reductions are considered.
\end{abstract}
\section{Introduction}
SDYM (ASDYM) equations
\bea
\mathbf{F}=\pm*\mathbf{F} 
\label{SDYM0}
\eea
represent SD (ASD) condition for the two-form
$\mathbf{F}=\d\mathbf{A}+\mathbf{A}\wedge\mathbf{A}$, where
the gauge field (potential) $\mathbf{A}$ is a one-form
taking its values in some Lie algebra.
The most well-known results concerning the integrability
of SDYM equations are formulated in four-dimensional 
Euclidean space or its compactification 
(sometimes complexification). However, 
it is a well-established fact in twistor theory 
that SDYM equations are integrable 
(in terms of twistor construction)
on general nontrivial geometrical
background defined by self-dual conformal structure 
on some 4-manifold, this structure itself 
being integrable \cite{Penrose}, \cite{Atiah}.
Different reductions of SDYM equations give rise to
background geometries which are themselves solutions of
(dispersionless) integrable systems, the current picture
of the field and many examples are provided in  
\cite{Calderbank}. Thus there is an approach 
to consider dispersionless integrable systems as 
integrable background geometries for some 
reductions of SDYM equations. 
We will take the opposite direction and, 
starting from dispersionless integrable hierarchies, 
will consider an extension leading in particular 
to SDYM equations on the self-dual backgroung. 
In the process of extension we will transfer all integrable 
structures of dispersionless integrable systems -- 
Lax pairs, the hierarchy, Lax-Sato equations, 
the dressing scheme \cite{BDM}, \cite{LVB09}
-- to the case of SDYM equations on the background.
Though some of these structures have their analogues in 
twistor approach, there are differences in technique and
setting of the problems, and comparison of these two approaches should be mutually enriching.
\section{ASD conformal structures}
Our starting point will be a recent result \cite{DFK2014}
\begin{theorem}[Dunajski, Ferapontov and Kruglikov (2014)]
\label{prop_ASD}
There exist local coordinates $(z, w, x, y)$ such 
that any ASD conformal structure
in signature $(2, 2)$  is locally represented by a metric 
\bea
\label{ASDmetric1}
{\textstyle\frac{1}{2}}g=dwdx-dzdy-F_y dw^2-(F_x-G_y)dwdz + G_xdz^2,
\eea
where the functions
$F,\;G:M^4\rightarrow \mathbb{R}$ satisfy a 
coupled system of third-order PDEs,
\bea
&&
\p_x(Q(F))+\p_y(Q(G))=0,
\nn\\
&&
(\p_w+F_y\p_x+G_y\p_y)Q(G)+(\p_z+F_x\p_x+G_x\p_y)Q(F)=0,
\label{sd_3rd}
\eea
where
\[
Q=\p_w\p_x-\p_z\p_y+F_y{\p_x}^2-G_x{\p_y}^2-
(F_x-G_y)\p_x\p_y.
\]
\end{theorem}
A conformal structure $[g]$  is called anti-self-dual
if the self-dual  part of the Weyl tensor of any $g\in[g]$ vanishes: $W_+=\frac12(W+*W)=0$.
Real case with the signature (2,2) or, generally, complex 
analytic case may be considered.  
The choice of ASD or SD case is just a convention 
depending on orientation, and we will
follow \cite{DFK2014} 
in taking ASD case 
(though we changed the order of 
variables and the sign of one variable for 
technical reasons).

A crucial observation made in \cite{DFK2014} is that
system (\ref{sd_3rd}) arises as $[X_1, X_2]=0$ from the dispersionless
Lax pair
\bea
X_1=\p_z-\lambda\p_x
+F_x\p_x+G_x\p_y+ f_1\p_\lambda,
\nn
\\
X_2=\p_w- \lambda\p_y + F_y\p_x+G_y\p_y+f_2\p_\lambda.
\label{new_lax}
\eea
Due to compatibility conditions, $f_1$ and $f_2$ can be 
expressed through $F$ and $G$,
\beaa
&&
f_1=-Q(G),\quad f_2=Q(F),
\\
&&
Q=\p_w\p_x-\p_z\p_y+F_y{\p_x}^2-G_x{\p_y}^2-(F_x-G_y)\p_x\p_y.
\eeaa

Theorem \ref{prop_ASD} defines an important relation
between ASD conformal structures and dispersionless
integrable systems. The Lax pair (\ref{new_lax}) 
corresponds to lowest order flows of generic 
four-dimensional dispersionless integrable hierarchy 
\cite{BDM},
\cite{LVB09}, 
we will briefly describe main structures of this 
hierarchy below.

On the other hand, reductions of the Lax pair (\ref{new_lax}) correspond to important geometrical systems. A reduction to divergence-free vector fields gives
the Dunajski system \cite{DP} describing null 
K\"ahler case,
and a further reduction to linearly degenerate case
($f_1,f_2=0$, no derivative over $\lambda$ in 
vector fields) 
leads to the famous Pleba\'nski second heavenly 
equation (Einstein ASD case).
\section{Extended Lax pair}
Let us consider an extension of Lax pair (\ref{new_lax})
to covariant derivatives,
\bea
&&
\nabla_{X_1}=\p_z-\lambda\p_x
+F_x\p_x+G_x\p_y+ f_1\p_\lambda + A_1,
\nn
\\
&&
\nabla_{X_2}=\p_w- \lambda\p_y + F_y\p_x+G_y\p_y+f_2\p_\lambda + A_2,
\label{extLax}
\eea 
where gauge field components $A_1$, $A_2$ do not 
depend on $\lambda$ and take their values in some (matrix)
Lie algebra.
Lax pairs of this structure 
(without derivative over $\lambda$) were already present 
in the seminal work of
Zakharov and Shabat \cite{ZS} (1979), where it was noticed
that the commutation relation splits into (scalar) vector
field part, which is the same as for unextended Lax pair,
and Lie algebraic part,
\beaa
&&
[\nabla_{X_1}, \nabla_{X_2}]=
[X_1, X_2]+ X_1 A_2 - X_2 A_1 +[A_1,A_2]=0 \Rightarrow
\\
&&
[X_1, X_2]=0, \quad \text{vector fields part}
\\
&&
X_1 A_2 - X_2 A_1 +[A_1,A_2]=0. \quad\text{matrix part}
\eeaa
From the first part we get system (\ref{sd_3rd}) describing
ASD conformal structure, and the
second part gives the system for $A_1$, $A_2$
\bea
&&
\partial_x A_2=\partial_y A_1,
\nn
\\
&&
(\p_z
+F_x\p_x+G_x\p_y)A_2 - 
(\p_w+ F_y\p_x+G_y\p_y)A_1 + [A_1,A_2]=0,
\label{SDYM}
\eea
which contains the coefficients of the metric
as some kind of `background'. 

For trivial background $F=G=0$ with conformal structure
\bea
g=2(dwdx-dzdy)
\label{constantmetric}
\eea
the extended Lax pair takes the form
\beaa
\nabla_{X_1}=\p_z-\lambda\p_x
+ A_1,
\quad
\nabla_{X_2}=\p_w- \lambda\p_y + A_2,
\eeaa 
and the commutativity condition is
\bea
\partial_x A_2=\partial_y A_1,
\quad
\p_z A_2 - 
\p_w A_1 + [A_1,A_2]=0, 
\label{SDYM00}
\eea
representing a well known  form
of ASDYM equations (\ref{SDYM0}) for constant metric $g$
(\ref{constantmetric})
in a special gauge (where two components of the gauge field
are eliminated by the gauge transform).

The following
statement demonstrates that general backgroud 
in equations (\ref{SDYM})
has a direct geometric sense.
\begin{theorem} 
\label{prop_ASDYM}
Equations (\ref{SDYM}) represent 
ASDYM equations (\ref{SDYM0})
for the background conformal structure (\ref{ASDmetric1}) (in a special gauge).
\end{theorem}
\begin{proof}
First we notice that for metric (\ref{ASDmetric1})
due to ASDYM (\ref{SDYM0}) 
equations we have
$
F_{34}=0,
$
where we have used the matrix inverse to metric $g$ defining
symmetric bivector
\bea
{\textstyle \frac{1}{2}}q
=\p_w\cdot\p_x-\p_z\cdot\p_y+ F_y\,\p_x^2+
(G_y-F_x)\,\p_x\cdot\p_y-G_x\,\p_y^2,
\label{bivector}
\eea
$\det{g}=\det q=1$ (for this metric $F^{12}=F_{34}$).
Then it is possible to choose a gauge such that
$
A_3=A_4=0,
$
and we have only two nontrivial gauge field components
$A_1$, $A_2$.

The next step is to introduce a tetrad of vector fields 
and dual tetrad
of one-forms for which the conformal 
structure (\ref{ASDmetric1}) takes especially 
simple form, see
\cite{DFK2014}.

In terms of tetrad of one-forms
\bea
\begin{array}{ll}
{\bf e}^{00'}=dw,& {\bf e}^{01'}=dy-G_y dw - G_x dz,
\\
{\bf e}^{10'}=dz, & {\bf e}^{11'}=dx-F_ydw-F_x dz,
\end{array}
\label{forms}
\eea
conformal structure (\ref{ASDmetric1})
is represented as \cite{DFK}
\beaa
g=2({\bf e}^{00'} {\bf e}^{11'}-{\bf e}^{01'} {\bf e}^{10'}),
\eeaa
where, following \cite{DFK2014}, we use spinor notations 
for 
indices.

The dual tetrad of vector fields is
\bea
\begin{array}{ll}
{\bf e}_{00'}=\p_w+F_y\p_x+G_y\p_y, 
&
{\bf e}_{01'}=\p_y
\\
{\bf e}_{10'}=\p_z+F_x\p_x+G_x\p_y, 
&
{\bf e}_{11'}=\p_x,
\end{array}
\label{fields}
\eea
and symmetric bivector (\ref{bivector}) reads 
\beaa
q=2({\bf e}_{00'} {\bf e}_{11'}-
{\bf e}_{01'} {\bf e}_{10'}).
\eeaa
ASDYM equations for this tetrad take the form
\beaa
F_{11'\;01'}=0, \quad
F_{00'\;10'}=0,\quad F_{00'\;11'}=F_{10'\;01'},
\eeaa
where, to calculate gauge field strength
$\mathbf{F}$ components in the tetrad basis,
we use a standard formula
\beaa
\mathbf{F}(\mathbf{u},\mathbf{v})=\nabla_\mathbf{u}\nabla_\mathbf{v}
-\nabla_\mathbf{v}\nabla_\mathbf{u}-
\nabla_{[\mathbf{u},\mathbf{v}]}
\eeaa
valid for arbitrary vector fields $\mathbf{u}$, $\mathbf{v}$.
Taking into account the structure of tetrade and the fact
that for our gauge $A_3=A_4=0$,  for the curvature components we get
\beaa
&&
F_{11'\;01'}=0,
\\
&&
F_{00'\;10'}=
(\p_w+F_y\p_x+G_y\p_y)A_1-(\p_z+F_x\p_x+G_x\p_y)A_2
-[A_1,A_2],
\\
&&
F_{00'\;11'}=-\p_x A_2,\quad
F_{10'\;01'}=-\p_y A_1
\eeaa
Thus ASDYM equations read
\beaa
&&
(\p_w+F_y\p_x+G_y\p_y)A_1-(\p_z+F_x\p_x+G_x\p_y)A_2
-[A_1,A_2]=0,
\\
&&
\p_x A_2=\p_y A_1,
\eeaa
that coincides with the Lie algebraic 
part of commutativity
condition for extended Lax pair (\ref{SDYM}).
\end{proof}
Theorem \ref{prop_ASD}, Theorem \ref{prop_ASDYM} imply the
following statement:
\begin{corollary}
Commutation relations $[\nabla_{X_1}, \nabla_{X_2}]=0$
for the Lax pair
\beaa
&&
\nabla_{X_1}=\p_z-\lambda\p_x
+F_x\p_x+G_x\p_y+ f_1\p_\lambda + A_1,
\nn
\\
&&
\nabla_{X_2}=\p_w- \lambda\p_y + F_y\p_x+G_y\p_y+f_2\p_\lambda + A_2,
\eeaa 
represent a general local form (up to coordinate and
gauge transformations) of ASDYM equations (\ref{SDYM0}) for
ASD conformal structure (real case, signature (2,2)).
\end{corollary}
Vector fields part of the commutation relation
defines equations (\ref{sd_3rd})
for the coefficients of the metric (\ref{ASDmetric1})
representing ASD conformal structure, and Lie-algebraic
part of commutation relations (\ref{SDYM}) gives ASDYM equations for this structure (in a special gauge).

The gauge freedom can be recovered by arbitrary gauge 
transformation, and it is possible to write
the Lax pair in gauge invariant form,
\bea
&&
\nabla_{X_1}=\p_z+
F_x\p_x+G_x\p_y + A_1-\lambda(\p_x+B_1)+ f_1\p_\lambda,
\nn\\
&&
\nabla_{X_2}=\p_w + F_y\p_x+G_y\p_y + A_2-\lambda (\p_y+B_2)+f_2\p_\lambda
\label{gaugeLax}
\eea
Vector field part of commutation relation remains 
the same and gives equations (\ref{sd_3rd}) for
conformal ASD structure.
Lie-algebraic part of commutation relations gives
the equations
\bea
&&
\partial_x B_2 - \partial_y B_1 + [B_1,B_2]=0,
\nn
\\
&&
\partial_x A_2 - (\p_w + F_y\p_x+G_y\p_y)B_1 
- [B_1,A_2]
\nn
\\
&&
\qquad=\partial_y A_1-(\p_z+ F_x\p_x+G_x\p_y)B_2
-[B_2,A_1],
\nn
\\
&&
(\p_z
+F_x\p_x+G_x\p_y)A_2 - 
(\p_w+ F_y\p_x+G_y\p_y)A_1 
\nn
\\
&&
\qquad
+ [A_1,A_2]
+ Q(F) B_1 +
Q(G) B_2
=0,
\label{gaugeASDYM}
\eea
representing gauge-invariant form of ASDYM equations (\ref{SDYM0})
for ASD conformal structure
(\ref{ASDmetric1}). Here matrices $A_1$, $A_2$, $B_1$,
$B_2$ are defined as components of the gauge field in
tetrad basis (\ref{forms}), (\ref{fields}) and can be
easily expressed through coordinate components.
Lax pair (\ref{gaugeLax}) may be considered as direct
covariant extension of the vector fields pair
(\ref{new_lax}) by changing partial derivatives
$\p_z$, $\p_w$, $\p_x$, $\p_y$ to covariant derivatives
$\nabla_z$, $\nabla_w$, $\nabla_x$, $\nabla_y$.
\section{Matrix dressing on the geometric background}
The dressing scheme for ASDYM equations based
on matrix Riemann-Hilbert (RH) problem 
(see, e.g., \cite{AT93})
can be extended to the case of nontrivial background.
Generally, we may consider matrix RH problem
\bea
\Phi_{+}=\Phi_{-}R(\psi_1,\psi_2,\psi_3),
\label{RH1}
\eea
defined on some oriented curve $\gamma$ in the complex plane,
or matrix $\dbar$ problem
\bea
\dbar\Phi=\Phi R(\psi_1,\psi_2,\psi_3),
\label{dbar1}
\eea
defined in some region $G$, where $\psi_i(\lambda,\mathbf{t})$ are arbitrary wave functions of 
dispersionless Lax pair
 \beaa
{X_1}\psi_i=(\p_z-\lambda\p_x
+F_x\p_x+G_x\p_y+ f_1\p_\lambda)\psi_i=0,
\\
{X_2}\psi_i=(\p_w- \lambda\p_y + F_y\p_x+G_y\p_y+f_2\p_\lambda)\psi_i=0,
\eeaa 
defined on $\gamma$ or in $G$. 
Let us suggest the existence of solution $\Phi$ of RH (or $\dbar$)
problem having no zeroes and normalized by 1 at infinity,
$
\Phi_\infty=1+\sum_{n=1}^\infty \Phi_n(\mathbf{t})\lambda^{-n}.
$
Then ${X_1}\Phi$, ${X_2}\Phi$ satisfy
the same problem ($[X_1,R]=[X_2,R]=0$), 
and the functions
$
({X_1}\Phi) \Phi^{-1},
$
$
({X_2}\Phi) \Phi^{-1}
$
are holomorphic in the complex plane.
Considering the
behaviour at infinity, we get
\beaa
({X_1}\Phi) \Phi^{-1}=-\partial_x\Phi_1(\mathbf{t}),
\\
({X_2}\Phi) \Phi^{-1}=-\partial_y\Phi_1(\mathbf{t}),
\eeaa
thus $\Phi$ is a solution for the extended Lax pair (\ref{extLax})
with the gauge field
\beaa
A_1=\partial_x\Phi_1(\mathbf{t}),\;
A_2=\partial_y\Phi_1(\mathbf{t}),
\eeaa
which satisfies equations (\ref{SDYM}).

Dropping the normalization condition at infinity and 
demanding only regularity,
we will get solution for gauge-invariant extended
Lax pair (\ref{gaugeLax}) and equations
(\ref{gaugeASDYM}).

For constant metric $g$ (\ref{constantmetric})
corresponding to trivial vector
fields we have wave functions
$\psi_1=\lambda$, $\psi_2=x+\lambda z$,
$\psi_3=y+\lambda w$,
and RH problem (\ref{RH1}) reduces 
to standard Riemann-Hilbert problem
for ASDYM equations (\ref{SDYM00}) (see, e.g., \cite{AT93}).

Important classes of reductions of equations (\ref{sd_3rd})
are connected with
the existence of wave functions 
for dispersionless Lax pair (\ref{new_lax})
with special analytic properties in $\lambda$, e.g., 
{\em polynomial wave functions}  $\psi=P^n(\lambda)$,
coefficients of the polynomial depends on times.
A class of special ASDYM solutions for these
background
geometries is defined by the problems
\beaa
\Phi_{+}=\Phi_{-}R(P^n)
\quad\text{or}\quad
\dbar\Phi=\Phi R(P^n).
\eeaa

Another important reduction is
{\em linearly-degenerate case}, for which 
there is no $\p_\lambda$ in dispersionless Lax pair and $\lambda$
is one of the wave functions, 
\beaa
\Phi_{+}=\Phi_{-}R(\lambda,\psi_1,\psi_2)
\quad\text{or}\quad
\dbar\Phi=\Phi R(\lambda,\bar\lambda,\psi_1,\psi_2).
\eeaa
In this case ASDYM Lax pair admits rational (in $\lambda$)
solutions with simple stationary poles (correspond
to $\delta$-functions in the $\dbar$ problem), which can
be calculated explicitly.
\section{From the dressing scheme to the hierarchy}
To introduce the hierarchy connected with extended
Lax pair (\ref{extLax}) 
and equations (\ref{SDYM}), we will start from
extended dressing scheme and obtain generating relations for the hierarchy and Lax-Sato equations. 
It is also possible 
to consider these generating relations independently 
in terms of formal series, and the dressing scheme as
a tool to construct solutions.

First we briefly outline the dressing scheme for 
multidimensional dispersionless hierarchy connected with
the Lax pair (\ref{new_lax}) \cite{BDM}, \cite{LVB09}. 
We consider nonlinear vector
Riemann-Hilbert problem on the unit circle,
\bea
&&
\Psi^0_\text{in}=F_0(\Psi^0_\text{out},\Psi^1_\text{out},\Psi^2_\text{out}),
\nn\\
&&
\Psi^1_\text{in}=F_1(\Psi^0_\text{out},\Psi^1_\text{out},\Psi^2_\text{out}),
\nn\\
&&
\Psi^2_\text{in}=F_2(\Psi^0_\text{out},\Psi^1_\text{out},\Psi^2_\text{out}),
\label{RHvector}
\eea
outside the unit circle the solutions are analytic and
given by expansions of the form
\beaa
&&
\Psi^0_\text{out}=\lambda+\sum_{n=1}^\infty \Psi^0_n(\mathbf{t}^1,\mathbf{t}^2)\lambda^{-n},
\\&&
\Psi^1_\text{out}=\sum_{n=0}^\infty t^1_n (\Psi^0)^{n}+
\sum_{n=1}^\infty \Psi^1_n(\mathbf{t}^1,\mathbf{t}^2)
\lambda^{-n}
\\&&
\Psi^2_\text{out}=\sum_{n=0}^\infty t^2_n (\Psi^0)^{n}+
\sum_{n=1}^\infty \Psi^2_n(\mathbf{t}^1,\mathbf{t}^2)
\lambda^{-n},
\eeaa 
inside the unit circle the functions are analytic.
Solutions to RH problem (\ref{RHvector})
$\Psi^0$, $\Psi^1$, $\Psi^2$ will give  wave fuctions for the
hierarchy of commuting vector fields,
defined through coefficients of expansion of these functions, $\mathbf{t}^1=(t^1_1,t^1_2,\dots)$,
$\mathbf{t}^2=(t^2_1,t^2_2,\dots)$ are two infinite sets of independent variables of the hierarchy.
To obtain a gauge field extension of the hierarchy,
we introduce also a matrix
Riemann-Hilbert problem
\beaa
\Phi_\text{in}=\Phi_\text{out}R(\Psi^0_\text{out},\Psi^1_\text{out},\Psi^2_\text{out}),
\eeaa
$\Phi$ is normalized by 1 at infinity and analytic inside
and outside
the unit circle,
\beaa
\Phi_\text{out}=1+\sum_{n=1}^\infty \Phi_n(\mathbf{t}^1,\mathbf{t}^2)\lambda^{-n}
\eeaa
Expansions of $\Psi$, $\Phi$ give coefficients for
extended Lax pair, $\Phi$ is a wave function. A general
wave function is given by the expression
$\Phi F(\Psi^0,\Psi^1,\Psi^2)$, $F$ is an arbitrary 
complex-analytic matrix function.

The vector fields part of the dressing scheme implies
analyticity in the complex plane of the form (no discontinuity on the unit circle)
\beaa
\omega=\left|\frac{D(\Psi^0,\Psi^1,\Psi^2)}{D(\lambda,x_1,x_2)}\right|^{-1}
\d\Psi^0\wedge\d\Psi^1\wedge\d\Psi^2,
\eeaa
where $x_1=t^1_0$, $x_2=t^2_0$ are 
lowest times of the hierarchy, and from matrix
Riemann problem we get analyticity of the
matrix-valued form
\beaa
\Omega=\omega\wedge\d\Phi\cdot \Phi^{-1}.
\eeaa


Analyticity of these forms imply the relations
\beaa
(\omega_\text{out})_-=\left(
\left|\frac{D(\Psi^0_\text{out},\Psi^1_\text{out},\Psi^2_\text{out})}{D(\lambda,x_1,x_2)}\right|^{-1}
\d\Psi^0_\text{out}\wedge\d\Psi^1_\text{out}\wedge\d
\Psi^2_\text{out}
\right)_-=0,
\eeaa
\beaa
(\Omega_\text{out})_{-}=(\omega_\text{out}
\wedge\d\Phi_\text{out}\cdot \Phi_\text{out}^{-1})_-=0
\eeaa
for the series $\Psi^0_\text{out}$, $\Psi^1_\text{out}$,
$\Psi^2_\text{out}$, $\Phi_\text{out}$.
These relations are generating relations for the hierarchy
in terms of formal series, they are equivalent to the complete set of Lax-Sato equations
of the hierarchy. Though we used the dressing scheme to
introduce these relations, they may be considered independently.

First relation gives Lax-Sato equations for the hierarchy
of commuting polynomial in $\lambda$ vector fields (here we drop subscript `out' for the series):
\bea
&&
\partial^k_n {\Psi}=\sum_{i=0}^2\left(
\left(\frac{D(\Psi^0,\Psi^1,\Psi^2)}{D(\lambda,x_1,x_2)}\right)^{-1}_{ik} 
(\Psi^0)^n\right)_+
{\partial_i}{\Psi},
\label{genSato}
\eea
where $\quad 1\leqslant n < \infty$, $k=1,2$,
$\p_0=\p_\lambda$, $\p_1=\p_{x_1}$, $\p_2=\p_{x_2}$,
$\Psi=(\Psi^0,\Psi^1,\Psi^2)$.


The second generating relation gives Lax-Sato
equations for $\Phi$ on the vector field background
in terms of extended polynomial vector fields,
\beaa
&&
\partial^k_n {\Psi}=V^k_n(\lambda)\Psi,
\\
&&
\partial^k_n {\Phi}=\left(V^k_n(\lambda)-
((V^k_n(\lambda)\Phi)\cdot\Phi^{-1})_+\right)\Phi,
\eeaa
where vector fields $V^k_n(\lambda)$ are defined 
by formula (\ref{genSato}).
First flows give exactly extended Lax pair
for ASDYM equations on ASD background (\ref{extLax}),
if we identify $z=t^1_1$, $w=t^2_1$, $x=x_1$, $y=x_2$.
\subsection*{Discussion, open questions}
Though ASD conformal structure in
Pleba\'nski-Robinson form (\ref{ASDmetric1}),
(\ref{bivector}) 
can be considered in general complex case, it is not convenient to use it
to obtain a real slice with Euclidean signature.
Motivated by the K\"ahler case,
we suggest to consider conformal structure
defined by the symmetric bivector
\bea
{\textstyle \frac{1}{2}}q=a\,\p_w\cdot\p_{\wt w}
+ b\,\p_z\cdot\p_{\wt w}+
c\,\p_w\cdot\p_{\wt z} + 
d\,\p_z\cdot\p_{\wt z} ,
\label{bivector1}
\eea
and corresponding extended Lax pair
\bea
&&
\nabla_{X_1}=\p_{\wt z}-\lambda(a\p_{ w}+b\p_{ z}) + 
(\lambda^2 f_1 + \lambda g_1)\p_{\lambda}
+A_1 -\lambda B_1,
\nn\\
&&
\nabla_{X_2}=\p_{\wt w} + \lambda(c\p_{ w}+d\p_{ z}) +
(\lambda^2 f_2 + \lambda g_2)\p_{\lambda} + A_2 + 
\lambda B_2, 
\label{KLax}
\eea
Vector fields part of commutation relations gives
seven equations for eight functions because of conformal
freedom. To fix representative of conformal structure and
close the system of equations,
it is convenient to use the condition
$
\det
\begin{pmatrix}
a&b
\\
c&d
\end{pmatrix}
=1,
$
in this case three independent coefficients of bivector
(\ref{bivector1}) satisfy three second-order equations,
and the conformal structure depends 
on 6 arbitrary functions of three variables.  
Another choice to fix conformal freedom is to put
$g_1=g_2=0$, 
in the scalar-flat K\"ahler case 
vector fields Lax pair
of this type was considered in \cite{Park}, \cite{Takasaki}.

The 
conjecture is that conformal structure (\ref{bivector1})
with the coefficients satisfying vector fields part of
commutation relations for the Lax pair (\ref{KLax})
gives a general local form of complex ASD conformal 
structure, and
Lie algebraic part of commutation relations gives
ASDYM equations on this background. 
The first part of the conjecture is connected 
with general recent results of the work \cite{CaldKrug}.
ASD conformal structure 
of the form (\ref{bivector1})
can be useful for 
reduction to Hermitian 
case and real case with Euclidean signature.

Another interesting question concerns covariant
extension of dispersionless integrable hierarchies 
in lower as well as in higher dimensions and the geometric
meaning of arising systems. An important
(2+1)-dimensional example is provided by the 
Manakov-Santini system. Let us consider extended Lax pair
\beaa
&&
\nabla_{X_1}=
\partial_y-(\lambda-v_{x})\partial_x + u_{x}\partial_\lambda + A,
\nn\\
&&
\nabla_{X_2}=\partial_t-(\lambda^2-v_{x}\lambda + u -v_{y})\partial_x
+(u_{x}\lambda+u_{y})\partial_\lambda +\lambda A + B,
\eeaa
where $A$, $B$ are gauge field components. Vector field
part of commutation relations gives the Manakov-Santini
system \cite{MS06}, \cite{MS07}
\bea
u_{xt} &=& u_{yy}+(uu_x)_x+v_xu_{xy}-u_{xx}v_y,
\nn\\
v_{xt} &=& v_{yy}+uv_{xx}+v_xv_{xy}-v_{xx}v_y,
\label{MSeq}
\eea
describing general Einstein-Weyl geometry \cite{DFK2014},
and matrix part of compatibility conditions read
\beaa
&&
A_y - B_x=0,
\\
&&
(\p_y+v_x\p_x)B-(\p_t+(v_y-u)\p_x) A + u_x A +[A,B]=0
\eeaa
For the potential $\Phi$, $A=\Phi_t$, $B=\Phi_y$
we  have
\bea
\Phi_{tx}-\Phi_{yy} - [\Phi_x,\Phi_y] - \p_x(u\Phi_x) + v_y \Phi_{xx} - v_x \Phi_{xy}=0,
\label{mono}
\eea
where $u$, $v$ satisfy Manakov-Santini system describing
Einstein-Weyl geometry.

The natural conjecture is that system (\ref{mono}) represents
a general local form of monopole equations on Einstein-Weyl
background (up to coordinate transformations and a gauge).
We are planning to consider extended Manakov-Santini hierarchy in more detail  in the nearest future.
\subsection*{Acknowledgments}
The main results of this article were first presented
at LMS -- EPSRC Durham Symposium
``Geometric and Algebraic Aspects of Integrability", 
and the author is grateful to the organisers
for kind hospitality.
The author also appreciates useful discussions with
E.V. Ferapontov, Macej Dunajski,
Boris Kruglikov and David Calderbank. This work 
was supported in part
by the President of Russia grant 9697.2016.2 (Scientific Schools).

\end{document}